\documentclass[aps,pra,reprint,superscriptaddress,nobibnotes,amsmath,amssymb,letterpaper]{revtex4-2}

\usepackage{mathtools}
\usepackage{amsthm}
\usepackage{xfrac}
\usepackage{graphicx}
\usepackage{dcolumn}
\usepackage{bm}

\newtheorem{axiom}{Principle}
\newtheorem{thm}{Theorem}
\newtheorem{cor}[thm]{Corollary}
\newtheorem{exmp}{Example}
\newtheorem{lemma}{Lemma}

\newcommand{\out}{u}
\newcommand{\outs}{\boldsymbol{u}}

\begin{document}

\preprint{APS/123-QED}

\title{
Contextuality analysis of disturbing data cannot preserve core Kochen--Specker principles
}

\author{Alisson Tezzin}
\email{alisson.cordeiro@ime.eb.br}
\affiliation{
 University of São Paulo, Institute of Mathematics and Computer Sciences, Av. Trab. São Carlense 400, 13566-590, São Carlos, SP, Brazil}

\author{Elie Wolfe}
\email{ewolfe@pitp.ca}
\affiliation{
 Perimeter Institute for Theoretical Physics
}

\author{Barbara Amaral}
\email{barbara\_amaral@usp.br}
\affiliation{
 University of São Paulo
}

\author{Matt Jones}
 \email{mcj@colorado.edu}
\affiliation{
 University of Colorado}

\date{\today}

\begin{abstract}
We address the challenge of extending the criterion for Kochen--Specker contextuality to data with disturbance, to accommodate experimental imperfections or applications to macroscopic settings beyond quantum physics.
We find that any such extension must abandon one or more principles that are core to the concept of contextuality.
Specifically, we prove that no extended criterion agreeing with Kochen--Specker contextuality in the absence of disturbance can simultaneously satisfy the following:
(1) any deterministic system is noncontextual; 
(2) discarding information cannot turn a noncontextual system into a contextual one; 
(3) classical post-processing cannot create contextuality; and 
(4) the joint realization of two statistically independent noncontextual systems is noncontextual.
From the standpoint of resource theories, the obstruction arises because free operations can turn noncontextual disturbing behaviors into contextual non-disturbing ones.
We suggest that future work on extended contextuality needs to move beyond the standard model-independent framework and account for the causal structure responsible for disturbance.
\end{abstract}

\maketitle

The Kochen--Specker theorem demonstrates the impossibility of assigning values to all observables of a quantum system in a way that preserves functional relations between them \cite{ kochenspecker67,isham1998toposI,abramsky2011sheaf, landsman2017foundations}. This result challenges realist interpretations of quantum theory where observables possess values at all times and obstructs the existence of noncontextual hidden-variable models for quantum mechanics \cite{kochenspecker67, isham1998toposI, doring2005vonNeumann}. Experimental tests suggest this is true also of real physical systems: Relative frequencies observed in different measurement contexts of identically prepared systems can exhibit \textit{contextuality}, meaning they are incompatible with any joint probability distribution \cite{budroni2022contextualityReview, bartosik2009contextuality, lapkiewicz2011indivisible, qu2020photons, leupold2018ion,aspect1981experimental, hensen2015loophole}.
Contextuality is also a valuable resource for numerous computational tasks \cite{um2013random, delfosse2015rebits, Howard2014ContextualityST, vega2017qubits, budroni2022contextualityReview}.

Standard frameworks for contextuality apply only under the assumption that measurements within a context do not directly influence each other. This assumption entails 
\textit{nondisturbance}, that overlapping contexts (i.e., sets of co-measured observables) must agree on probabilities in their intersection \cite{abramsky2011sheaf,ramanathan2012generalized,popescu1994axiom,gleason1957measures},
as well as \textit{compatibility}, that if measurements within a context are performed sequentially the order does not matter \cite{guhne2010compatibility,larsson2011violating}.
These assumptions are sometimes justified by spacelike separation in Bell scenarios \cite{bell1964epr} or by no-fine-tuning principles in general contextuality scenarios \cite{wood2015tuning,cavalcanti2018tuning,pearl2021classical}.
However, when disturbance is observed, there is manifestly some direct influence between measurements and it does not make sense to restrict classical explanations to strictly noncontextual models \cite{Jones2019mcontextuality}. 

Unfortunately, most real experiments exhibit disturbance 
due to unavoidable imperfections in preparation or measurement apparatus and statistics of finite samples \cite{kirchmair2009state,adenier2006anomalies,adenier2017test,kujala2015extendedNoncontextuality,budroni2022contextualityReview}.
Moreover, there has been much recent interest in applying the concept of contextuality to fields beyond quantum physics, in macroscopic settings where disturbance is the rule rather than the exception 
\cite{dzhafarov2017cbd2.0, dzhafarov2017canonical, aerts2014quantum,wang2021analysing, mansfield2024developments, lo2022model, wang2021quantumlike}.

To enable contextuality analysis in the presence of disturbance, several proposals for extended criteria for contextuality relax the assumptions of hidden-variable models while still aiming to draw no-go conclusions.
These include allowing causal effects between measurements or from contexts to measurement outcomes
\cite{kirchmair2009state, guhne2010compatibility, winter2014does, kujala2015extendedNoncontextuality, vallee2024corrected}, 
allowing the model's predictions to be unconstrained in some proportion of runs \cite{Nagali, Marques}, or 
allowing nondeterministic outcomes conditioned on the hidden state \cite{vallee2024corrected}.
These relaxations lead to extended noncontextuality inequalities 
that apply to data with disturbance
and that reduce to standard Kochen--Specker inequalities when disturbance is absent.

A conceptual difficulty with all these approaches is that they require auxiliary assumptions to bound the magnitude of the relevant assumption violation (which is necessarily unobservable). 
Specifically, they limit the violation to the minimum required to account for certain marginal statistics of the data, including
differences in marginal expectations between contexts \cite{kujala2015extendedNoncontextuality, vallee2024corrected}
or between orderings of sequential measurements \cite{guhne2010compatibility, Jerger},
or violation of repeatability \cite{kirchmair2009state, guhne2010compatibility, vallee2024corrected} or exclusivity \cite{Nagali, D'Ambrosio, Marques}.
This choice implies the most liberal possible criterion for contextuality.
This is a tenuous position, because 
once a violation is known to be present, 
it is hard to argue that it could not be stronger than the marginal statistics reveal \cite{atmanspacher2019contextuality,yearsley2019contextuality,Jones2019mcontextuality,vallee2024corrected}.
Nevertheless, researchers have employed these extended theories to
certify and quantify contextuality and nonlocality in a variety of theoretical and experimental settings both in and outside physics
\cite{bacciagaluppi2014leggett,kujala2015extendedNoncontextuality,zhan2017experimental,fluhmann2018sequential, malinowski2018probing, kupczynski2021contextuality, wang2022significant,khrennikov2022complementarity, kirchmair2009state, Marques, Nagali, Jerger, D'Ambrosio, dzhafarov2016contextuality,cervantes2018snow,basieva2019true}.
Considering the strong physical (and metaphysical) conclusions regarding classical hidden-variable accounts often derived from these results
\cite{fluhmann2018sequential,malinowski2018probing,wang2022significant,kirchmair2009state,Jerger,zhan2017experimental}, 
understanding the logical implications of extended criteria for contextuality is a pressing issue.

In this Letter, rather than debating the specific assumptions of existing proposals, 
we take a more general axiomatic approach.
Specifically, we consider four principles that (as we argue below) capture essential properties of KS contextuality, 
and prove that no extension to disturbing systems can satisfy all of them.
The first two principles are adopted from previous work following a similar line of investigation \cite{kujala2021dichotomization}:
{\bf Determinism} is the principle that deterministic systems are noncontextual, and
{\bf Monotonicity} is the principle that disregarding information from a noncontextual system cannot make it contextual.
In addition, we propose principles of {\bf Post-processing}, that classical computation on the outcomes of jointly measured observables cannot create contextuality, and {\bf Independence}, that the joint realization of two statistically independent noncontextual systems is noncontextual. 
Monotonicity, Post-processing, and Independence 
are self-consistency properties whereby if one behavior is noncontextual then so must be certain others that can be obtained by classical transformations of the data. 
This accords with resource theories of contextuality, in which these transformations
play an important role as free operations that cannot create or increase contextuality \cite{shane2019comonadic,Karvonen2019categories,barbosa2021closing}.

We first formalize these principles and then
prove in Theorem \ref{thm:KS-satisfies-axioms} that KS contextuality satisfies them all. 
However, we prove in Theorem \ref{thm:impossibility-a} that if one attempts to extend the criterion for contextuality beyond nondisturbing systems to include disturbing ones, the principles cannot be jointly satisfied.

We follow the ``model-independent'' formulation of contextuality used in most modern work, which generalizes the notions of measurement and compatibility beyond quantum mechanics \cite{abramsky2011sheaf,acin2015combinatorial,CSW, amaral2018graph,kujala2015extendedNoncontextuality}. 
Our mathematical setup mirrors this previous work \cite{abramsky2011sheaf,kujala2015extendedNoncontextuality}. A \textit{finite measurement scenario} is a quadruple $\mathcal{S} \equiv (\mathcal{Q}, \mathcal{C}, \prec,  \mathcal{O})$, where $\mathcal{Q}$ and $\mathcal{C}$ are finite sets, $\prec$ is a relation in $\mathcal{Q} \times \mathcal{C}$, and $\mathcal{O} \equiv (O_{q})_{q \in \mathcal{Q}}$ is a family of finite sets. We say $\mathcal{Q}$ is a set of observables, $\mathcal{C}$ is a set of contexts, and $q\prec c$ indicates that observable $q$ is measured in context $c$.
$O_{q}$ represents the set of possible outcomes for observable $q$, while $O^{c} \equiv \prod_{q \prec c}O_{q}$ is the set of possible joint outcomes for $\boldsymbol{c} \equiv \{ q \in \mathcal{Q}: q \prec c \}$. 
We focus on scenarios with finitely many observables, contexts and outcomes, because these are sufficient to prove our main result. More general definitions are in Ref. \cite{Jones2019mcontextuality} and elsewhere. 

In KS contextuality a context is defined as any set of compatible (i.e., co-measurable) observables 
and therefore can be identified with the set of observables measured in it ($c\equiv\boldsymbol{c}$). 
Nondisturbance (as defined below) implies any two contexts with the same set of observables must be empirically indistinguishable. 
However, when extending contextuality to disturbing behaviors it can be useful to allow distinct contexts to contain the same observables (i.e., $c\ne c'$ with $\boldsymbol{c}=\boldsymbol{c'}$) when they are distinguished by some other experimental factor such as the order in which the measurements are made \cite{guhne2010compatibility, Jerger,wang2014judments,Dzhafarov2023sheaf}.
Therefore we adopt a broader definition of context as any well-defined and replicable condition under which an experimenter measures one or more observables,
and we define $\mathcal{C}$ as an abstract set of labels for these conditions.
Importantly, the two definitions of context are equivalent because one can always append a unique dummy observable to each context to encode the value of any relevant experimental factor(s).

A \textit{behavior} on $\mathcal{S}$ is a family of probability distributions $P(\cdot| c)$ on $O^{c}$ for all $c\in\mathcal{C}$. For any $\outs \in O^{c}$, $P(\outs|c)$ is the probability of obtaining joint outcome $\outs$ in context $c$. For any collection of observables $\boldsymbol{q} \subset \boldsymbol{c}$, $P_{\boldsymbol{q}}(\cdot\vert c)$ denotes the marginal of $P(\cdot \vert c)$ on $O_{\boldsymbol{q}} \equiv \prod_{q \in \boldsymbol{q}} O_{q}$. If $\boldsymbol{q}$ is a singleton $\{q\}$, we write $P_{q}(\cdot\vert c)$ instead.  Other shorthands should be self-explanatory, such as  $P_{(q_{1},q_{2})}(\out_{1},\out_{2}\vert c)$ to express joint probabilities for $\{q_{1},q_{2}\}$ when $\out_{1}\in O_{q_{1}}$ and $\out_{2}\in O_{q_{2}}$.

\begin{exmp}[Behavior]\label{ex: behavior} Let $\mathcal{S}$ be a scenario given by two binary observables $q_{1},q_{2}$ and four contexts $c_{1},\dots,c_{4}$ that contain them both.
Let $0\le p\le1$.
The following tables represent a behavior $P^{(1)}$ on $\mathcal{S}$, where the entry in row $q_{1}=i$ and column $q_{2}=j$ of the table $P^{(1)}( \cdot |c_{k})$ represents the probability $P^{(1)}_{(q_{1},q_{2})}(i,j|c_{k})$.

\vspace{2mm}
\begin{tabular}{c|cc c c|cc}
$\boldsymbol{P^{(1)}(\cdot|c_{1})}$ & $q_{2}=0$ & $q_{2}=1$ & \; & 
$\boldsymbol{P^{(1)}(\cdot|c_{2})}$ & $q_{2}=0$ & $q_{2}=1$ \tabularnewline
\cline{1-3} \cline{2-3} \cline{3-3} \cline{5-7} \cline{6-7} \cline{7-7} 
$q_{1}=0$ & $\frac{1-p}{2}$ & $\frac{p}{2}$ &  & 
$q_{1}=0$ & $\frac{1-p}{2}$ & $\frac{p}{2}$ \tabularnewline
$q_{1}=1$ & $\frac{1-p}{2}$ & $\frac{p}{2}$ &  & 
$q_{1}=1$ & $\frac{1-p}{2}$ & $\frac{p}{2}$ \tabularnewline
\end{tabular}

\vspace{2mm}
\begin{tabular}{c|cc c c|cc}
$\boldsymbol{P^{(1)}(\cdot|c_{3})}$ & $q_{2}=0$ & $q_{2}=1$ & \; & 
$\boldsymbol{P^{(1)}(\cdot|c_{4})}$ & $q_{2}=0$ & $q_{2}=1$\tabularnewline
\cline{1-3} \cline{2-3} \cline{3-3} \cline{5-7} \cline{6-7} \cline{7-7} 
$q_{1}=0$ & $\frac{1-p}{2}$ & $\frac{p}{2}$ &  & 
$q_{1}=0$ & $\frac{p}{2}$ & $\frac{1-p}{2}$\tabularnewline
$q_{1}=1$ & $\frac{1-p}{2}$ & $\frac{p}{2}$ &  & 
$q_{1}=1$ & $\frac{p}{2}$ & $\frac{1-p}{2}$\tabularnewline
\end{tabular}

\end{exmp}


A behavior $P$ is \textit{nondisturbing} if $P_{\boldsymbol{q}}(\cdot |c) = P_{\boldsymbol{q}}(\cdot |c')$ whenever $\boldsymbol{q} \subset \boldsymbol{c} \cap \boldsymbol{c'}$ \cite{abramsky2011sheaf, amaral2018graph}. Otherwise, $P$ is \textit{disturbing}.
The behavior in Example \ref{ex: behavior} is disturbing because $P^{(1)}_{q_2}(\cdot|c_4) \ne P^{(1)}_{q_2}(\cdot|c_i)$ for $i=1,2,3$ (assuming $p\ne\sfrac{1}{2}$).
Notice this definition can encompass order effects as often used to test compatibility with sequential measurements \cite{guhne2010compatibility,Marques,Jerger}, by treating different orderings as distinct contexts.

A nondisturbing behavior $P$ is \textit{KS noncontextual} if there exists a global probability assignment for it, namely a probability distribution $\overline{P}$ on $\prod_{q \in \mathcal{Q}}O_{q}$ whose marginals match $P$ in each context: $\overline{P}_{\boldsymbol{c}}(\cdot)=P(\cdot|c)$ for all $c\in\mathcal{C}$ \cite{fine1982hidden,abramsky2011sheaf, amaral2018graph}. Otherwise, $P$ is KS-contextual. 
Thus, KS noncontextuality formalizes the realist notion that all observables are (probabilistically) well-defined regardless of which subset are measured.

An extended criterion for contextuality is
any method for determining whether a behavior, including disturbing ones, is or is not contextual. Formally, it amounts to what we define as an \textit{extension of contextuality}:
a partition that labels every behavior as either contextual or noncontextual and that agrees with KS contextuality for nondisturbing behaviors. 

Following the model-independent view of contextuality, we assume an extension of contextuality depends only on the relationships among the contexts, observables, and outcome distributions. Thus we require the contextuality status of a behavior to be unchanged if it is redefined on a scenario that is isomorphic to the original one (by renaming contexts, observables, or outcomes). 
This is a standard assumption \cite{abramsky2011sheaf, amaral2018graph, kujala2015extendedNoncontextuality, budroni2022contextualityReview}, although we suggest below that one way to escape our main result may be for an extension of contextuality to depend on the physical (or causal) arrangements of measurement systems.

We next formalize the four principles proposed above and offer intuitive arguments that they are central to the original notion of contextuality.

{\bf Determinism} can be motivated from quantum mechanics, where contextuality is possible only for incompatible observables \cite{xu2019necessary,khrennikov2021can}, which in turn must be nondeterministic (since a deterministic quantum observable commutes with all other observables). 
Likewise, causal interpretations of contextuality treat contextuality as requiring causal influence that is fine-tuned or otherwise hidden in observable distributions \cite{cavalcanti2018tuning,Jones2019mcontextuality,wood2015tuning}. This is not possible in a deterministic system, where all causal influences are fully observable. 

A behavior $P$ is deterministic if for each $c$ there is a $\outs_c \in O^{c}$ satisfying $P(\outs_c|c) = 1$. Equivalently, given any $q\prec c$, there exists $\out_{q,c}\in O_q$ with $P_q(\out_{q,c}|c) = 1$ \footnote{This should not be confused with a deterministic model, in which outcomes are deterministic conditioned on stochastic hidden states \cite{abramsky2011sheaf}.}.

\begin{axiom}[Determinism]\label{ax: determinism}
Any deterministic behavior is noncontextual.
\end{axiom}

{\bf Monotonicity} can be motivated from contextuality's roots in the Einstein--Podolsky--Rosen question of the completeness of quantum mechanics \cite{einstein1935incomplete}, later formalized by Bell \cite{bell1964epr} and Kochen and Specker \cite{kochenspecker67} as a question of whether correlations between measurements can be explained by unobserved information. 
Contextuality embodies the fact that the answer to this question is negative.
Therefore measuring more information from a contextual system should never make it noncontextual. Equivalently, discarding information from a noncontextual system should never make it contextual.

Although monotonicity can take multiple forms (see coarse-graining in Appendix \ref{appendix: relations}), we focus on removing information by dropping certain observables from certain contexts \cite{kujala2021dichotomization}. 
Formally, assume
scenario $\mathcal{S} = \{\mathcal{Q}, \mathcal{C}, {\prec}, \mathcal{O}\}$
contains
$\mathcal{S}' = \{\mathcal{Q}', \mathcal{C}', \prec', \mathcal{O}'\}$, meaning
$\mathcal{Q}'\subseteq\mathcal{Q}$,
$\mathcal{C}'\subseteq\mathcal{C}$,
$q\prec' c$ only if $q\prec c$,
and $O'_q=O_q$ for all $q\in\mathcal{Q}'$.
Given a behavior $P$ on $\mathcal{S}$, the \emph{restriction} $P'$ to $\mathcal{S'}$ is defined by marginalizing $P(\cdot|c)$ for every $c\in\mathcal{C}'$: $P'(\cdot|c) = P_{\boldsymbol{c}'}(\cdot|c)$ where $\boldsymbol{c}'=\{q\in\mathcal{Q}':q\prec'c\}$. 
Conversely if $P'$ is a restriction of $P$, $P$ is said to be an \emph{expansion} of $P'$.

\begin{axiom}[Monotonicity \cite{kujala2021dichotomization}]\label{ax: Nestedness}
Any restriction of a noncontextual behavior is noncontextual.
\end{axiom}

{\bf Post-processing} can be motivated from the original conception of contextuality as a property that cannot be explained by certain classical theories, specifically hidden-variable models in which each observable is a deterministic function of some latent state \citep{bell1964epr,kochenspecker67}. This suggests that classical computation, yielding new observables that are deterministic functions of the existing ones, cannot create contextuality.

We define a post-processing of a collection of observables $\boldsymbol{q}=\{q_{1}\dots,q_n\}$ as a deterministic function $f(q_{1}, \dots, q_{n})$. This yields a new observable $q' \equiv f(\boldsymbol{q})$ with $q'\prec c$ whenever $\boldsymbol{q}\prec c$. 
Given a behavior $P$, the expanded behavior $P'$ incorporating $q'$ is defined by $\mathcal{Q}'=\mathcal{Q}\cup\{q'\}$ with $P'_{(\boldsymbol{q},q')}(\boldsymbol{\out},\out'|c)$ equal to $P_{\boldsymbol{q}}(\boldsymbol{\out}|c)$ if $\out'=f(\boldsymbol{\out})$ and zero otherwise
(i.e., $q'=f(q_1,\dots,q_n)$ with probability $1$). 

\begin{axiom}[Post-processing]
\label{ax: post-processing}
Given any noncontextual behavior and a post-processing $f(q_{1},\dots q_{n})$ of observables in that scenario, the expanded behavior obtained by incorporating $q'\equiv f(q_1,\dots,q_n)$ is also noncontextual.
\end{axiom}

A further motivation for Principle \ref{ax: post-processing} is the view of contextuality as an objective property of physical systems independent of how experimenters encode their measurements, when different encodings are logically equivalent. This is illustrated in the following example.

\begin{exmp}[Logically equivalent encodings]\label{ex: post-processing}  Given $P^{(1)}$ from Example \ref{ex: behavior}, define the post-processing 
$q_{3}\equiv I_{\{q_1=q_2\}}$.
That is, $q_{3}=1$ if $q_{1}=q_{2}$ and $q_{3}=0$ if $q_{1}\ne q_{2}$. Incorporating $q_3$ and then dropping $q_2$ yields the following behavior $P^{(2)}$:

\vspace{2mm}
\begin{tabular}{c|cc c c|cc}
$\boldsymbol{P^{(2)}(\cdot|c_{1})}$ & $q_{3}=0$ & $q_{3}=1$ & \; & 
$\boldsymbol{P^{(2)}(\cdot|c_{2})}$ & $q_{3}=0$ & $q_{3}=1$ \tabularnewline
\cline{1-3} \cline{2-3} \cline{3-3} \cline{5-7} \cline{6-7} \cline{7-7} 
$q_{1}=0$ & $\frac{p}{2}$ & $\frac{1-p}{2}$ &  & 
$q_{1}=0$ & $\frac{p}{2}$ & $\frac{1-p}{2}$ \tabularnewline
$q_{1}=1$ & $\frac{1-p}{2}$ & $\frac{p}{2}$ &  & 
$q_{1}=1$ & $\frac{1-p}{2}$ & $\frac{p}{2}$ \tabularnewline
\end{tabular}
\vspace{2mm}

\begin{tabular}{c|cc c c|cc}
$\boldsymbol{P^{(2)}(\cdot|c_{3})}$ & $q_{3}=0$ & $q_{3}=1$ & \; & 
$\boldsymbol{P^{(2)}(\cdot|c_{4})}$ & $q_{3}=0$ & $q_{3}=1$\tabularnewline 
\cline{1-3} \cline{2-3} \cline{3-3} \cline{5-7} \cline{6-7} \cline{7-7} 
$q_{1}=0$ & $\frac{p}{2}$ & $\frac{1-p}{2}$ &  & 
$q_{1}=0$ & $\frac{1-p}{2}$ & $\frac{p}{2}$\tabularnewline
$q_{1}=1$ & $\frac{1-p}{2}$ & $\frac{p}{2}$ &  & 
$q_{1}=1$ & $\frac{p}{2}$ & $\frac{1-p}{2}$\tabularnewline
\end{tabular}\medskip{}
\vspace{2mm}

Conversely, $P^{(1)}$ can be obtained from $P^{(2)}$ by incorporating the post-processing 
$q_2 \equiv I_{q_1=q_3}$
and then dropping $q_3$. Thus Principles \ref{ax: Nestedness} and \ref{ax: post-processing} entail that if $P^{(1)}$ is noncontextual then so is $P^{(2)}$, and vice versa. This is reasonable because $(q_1,q_2)$ and $(q_1,q_3)$ are logically equivalent ways of representing the same observed information.

\end{exmp}

{\bf Independence} can be motivated from the common description of contextuality as the combination of local consistency and global inconsistency \cite{abramsky2011sheaf, barbosa2021closing}. 
If two subsystems are each globally consistent (noncontextual), and they co-exist but are statistically independent, then they should also be globally consistent when taken together. The independent union of systems was introduced by Abramsky et al.\ \cite{abramsky2017fraction} as a free operation that cannot create or increase contextuality. It was also used by Cabello to derive a bound on nondisturbing contextual behaviors that coincides with the quantum bound \cite{cabello2019simple}.
Independence can be assumed to arise whenever the systems are jointly measurable but have negligible common history, as with experiments in separate labs by communicating observers (excluding superdeterminism). 

Let $\mathcal{S}_{1}$ and $\mathcal{S}_{2}$ be disjoint scenarios, so that $\mathcal{Q}_{1}\cap \mathcal{Q}_{2} = \emptyset = \mathcal{C}_{1} \cap \mathcal{C}_{2}$. We define the \emph{product scenario} $\mathcal{S}_1\otimes\mathcal{S}_2$ by taking the union of observables $\mathcal{Q}_1\cup\mathcal{Q}_2$ and the Cartesian product of contexts $\mathcal{C}_1\times\mathcal{C}_2$. Each context $(c_1,c_2)$ is a pair of contexts for the two subsystems, with $q\prec(c_1,c_2)$ iff $q\prec c_1$ (for $q\in\mathcal{Q}_1$) or $q\prec c_2$ (for $q\in\mathcal{Q}_2$).
Given behaviors $P^{(1)}$ on $\mathcal{S}_1$ and $P^{(2)}$ on $\mathcal{S}_2$, the \emph{product behavior} $P^{(1)}\otimes P^{(2)}$ is defined by assuming the two systems are statistically independent in every joint context: $(P^{(1)}\otimes P^{(2)})_{(\boldsymbol{q_1},\boldsymbol{q_2})}(\outs_1,\outs_2|(c_1,c_2)) = P^{(1)}_{\boldsymbol{q_1}}(\outs_1|c_1)P^{(2)}_{\boldsymbol{q_2}}(\outs_2|c_2)$ for all collections of observables $\boldsymbol{q_1}\prec c_1$ and $\boldsymbol{q_2}\prec c_2$ and all joint outcomes $\outs_1$ and $\outs_2$. Thus $P^{(1)} \otimes P^{(2)}$ represents the joint realization of two statistically independent behaviors $P^{(1)}$ and $P^{(2)}$.  

\begin{axiom}[Independence]\label{ax: independence} The product of noncontextual behaviors is noncontextual.
\end{axiom}

Although we have written the principles in terms of behaviors, they can be interpreted equivalently in terms of no-go conclusions for classical models.
Whatever family of relaxed hidden-variable models one adopts for analyzing contextuality with disturbance, our proposed principles require it to satisfy the following closure properties:
Any deterministic behavior has a model (Principle \ref{ax: determinism}), any model of a behavior can be collapsed to model a restriction (Principle \ref{ax: Nestedness}) or extended to model a post-processing (Principle \ref{ax: post-processing}) of that behavior, and any two models of independent behaviors can be combined to model them jointly (Principle \ref{ax: independence}).
Each of these properties involves a simple, classical model transformation (or for Determinism, a trivial model with one hidden state).
As we show in Theorem \ref{thm:KS-satisfies-axioms}, they all hold for KS-noncontextual models (see Appendix \ref{appendix: theorem-proof} for proof).

\begin{thm}
\label{thm:KS-satisfies-axioms}
KS contextuality (defined for nondisturbing behaviors) satisfies
Determinism, Monotonicity, Post-Processing and Independence (Principles \ref{ax: determinism}, \ref{ax: Nestedness}, \ref{ax: post-processing} \& \ref{ax: independence}).
\end{thm}

Our main result is that the corresponding statement necessarily fails if we allow for disturbance: no extension of contextuality to disturbing systems is consistent with the proposed principles. 
The proof begins with two trivially noncontextual behaviors---one nondisturbing and KS-noncontextual, and the other disturbing and deterministic---and applies the free operations entailed by the other principles to obtain a KS-contextual behavior (quantum predictions violating the CHSH inequality \cite{clauser1969proposed}).
We further show that the inconsistency in this proof applies to all other behaviors (see Corollary \ref{cor: generality} in Appendix \ref{appendix: generality}).

\begin{thm}
\label{thm:impossibility-a}There is no extension of contextuality to disturbing behaviors that satisfies Determinism, Monotonicity, Post-processing and Independence (Principles \ref{ax: determinism}, \ref{ax: Nestedness}, \ref{ax: post-processing} \& \ref{ax: independence}).
\end{thm}
\begin{proof}

Let $P^{\rm A}$ contain binary observables $q_1$ and $q_{\rm A}$ measured in a single context (which we suppress in our notation) with independent probabilities $P^{\rm A}_{q_1}(1)=\sfrac{1}{2}$ and $P^{\rm A}_{q_{\rm A}}(1)=p$.
Let $P^{\rm B}$ be a deterministic behavior with a single binary observable $q_{\rm B}$ measured in four contexts, with $P^{\rm B}_{q_{\rm B}}(1|c_{1}) = P^{\rm B}_{q_{\rm B}}(1|c_{2})=P^{\rm B}_{q_{\rm B}}(1|c_{3})=P^{\rm B}_{q_{\rm B}}(0|c_{4}) = 1$
($c_1\dots c_4$  could be defined by co-measuring $q_{\rm B}$ with ancillary observables $q_1^+\dots q_4^+$, which we suppress).
$P^{\rm A}$ is KS-noncontextual, and $P^{\rm B}$ is noncontextual by Principle \ref{ax: determinism}.

Define the product behavior $P^{\rm A}\otimes P^{\rm B}$ and the post-processing $q_2\equiv I_{q_{\rm A}=q_{\rm B}}$. 
Then drop $q_{\rm A}$ and $q_{\rm B}$.
The result is $P^{(1)}$ from Example \ref{ex: behavior}. Therefore Principles \ref{ax: independence}, \ref{ax: post-processing} and \ref{ax: Nestedness} imply $P^{(1)}$ is also noncontextual. 
Next define the post-processing $q_{3}\equiv I_{q_1=q_2}$.
Incorporating $q_3$ and dropping $q_{2}$ yields $P^{(2)}$ from Example \ref{ex: post-processing}, which is noncontextual by Principles \ref{ax: Nestedness} and \ref{ax: post-processing}. 

Finally, relabel $q_{1}$ as $q_{4}$ in $c_{3}$ and $c_{4}$, and relabel $q_{3}$ as $q_{5}$ in $c_{2}$ and $c_{4}$. This relabeling can be done by post-processing with $f$ as the identity function ($q_4\equiv q_1$, $q_5\equiv q_3$) followed by restriction (dropping $q_1$ from $c_3,c_4$ and $q_3$ from $c_2,c_4$) (see Lemma \ref{lemma: N+P>R} in Appendix \ref{appendix: relations}). 
If we choose $p=(2+\sqrt{2})/4$ then
the result is the following behavior $P^{(3)}$, which must be noncontextual by Principles \ref{ax: Nestedness} and \ref{ax: post-processing}:

\vspace{2mm}
\begin{tabular}{c|cc c c|cc}
$\boldsymbol{P^{(3)}(\cdot|c_{1})}$ & $q_{3}=0$ & $q_{3}=1$ & \; & 
$\boldsymbol{P^{(3)}(\cdot|c_{2})}$ & $q_{5}=0$ & $q_{5}=1$ \tabularnewline
\cline{1-3} \cline{2-3} \cline{3-3} \cline{5-7} \cline{6-7} \cline{7-7} 
$q_{1}=0$ & $\frac{2+\sqrt{2}}{8}$ & $\frac{2-\sqrt{2}}{8}$ &  & 
$q_{1}=0$ & $\frac{2+\sqrt{2}}{8}$ & $\frac{2-\sqrt{2}}{8}$ \tabularnewline
$q_{1}=1$ & $\frac{2-\sqrt{2}}{8}$ & $\frac{2+\sqrt{2}}{8}$ &  & 
$q_{1}=1$ & $\frac{2-\sqrt{2}}{8}$ & $\frac{2+\sqrt{2}}{8}$ \tabularnewline
\end{tabular}

\vspace{2mm}
\begin{tabular}{c|cc c c|cc}
$\boldsymbol{P^{(3)}(\cdot|c_{3})}$ & $q_{3}=0$ & $q_{3}=1$ & \; & 
$\boldsymbol{P^{(3)}(\cdot|c_{4})}$ & $q_{5}=0$ & $q_{5}=1$\tabularnewline
\cline{1-3} \cline{2-3} \cline{3-3} \cline{5-7} \cline{6-7} \cline{7-7} 
$q_{4}=0$ & $\frac{2+\sqrt{2}}{8}$ & $\frac{2-\sqrt{2}}{8}$ &  & 
$q_{4}=0$ & $\frac{2-\sqrt{2}}{8}$ & $\frac{2+\sqrt{2}}{8}$\tabularnewline
$q_{4}=1$ & $\frac{2-\sqrt{2}}{8}$ & $\frac{2+\sqrt{2}}{8}$ &  & 
$q_{4}=1$ & $\frac{2+\sqrt{2}}{8}$ & $\frac{2-\sqrt{2}}{8}$\tabularnewline
\end{tabular}
\vspace{2mm}

\noindent However, this is 
the maximal quantum violation of the CHSH inequality \cite{clauser1969proposed}, 
which is a KS-contextual nondisturbing behavior \cite{basu2001bell,araujo2013all}. Thus we have a contradiction.
\end{proof}

Note that this proof uses only the milder form of 
post-processing mentioned in Example \ref{ex: post-processing}, where we merely re-encode the observables using a lossless transformation $(q_1, q_2) \leftrightarrow (q_1, q_3)$.

Theorem \ref{thm:impossibility-a} implies any method of contextuality analysis for disturbing behaviors must abandon one or more properties that we have argued are central to KS contextuality.
This is reminiscent of ``impossibility theorems'' in economic decision theory showing that a set of desirable axioms cannot be simultaneously satisfied \citep{arrow1950difficulty,geanakoplos2005three}.
Specifically, any extension of contextuality must label a deterministic behavior as contextual or it must be inconsistent by labeling a behavior as noncontextual while labeling another---derived from the first by restriction, classical post-processing, or independent combination---as contextual.
Moreover, this inconsistency is not limited to the behaviors used in our proof. Starting with any behavior that is classified as noncontextual, one can use Principle \ref{ax: independence} to append the construction in the proof and proceed to a behavior containing $P^{(3)}$. Thus the contrapositive of our argument implies that 
no behavior can be consistently labeled noncontextual if contextuality is extended to include disturbing behaviors (see Appendix \ref{appendix: generality} for proof).

Although we defined KS contextuality in the more general sheaf-theoretic framework \cite{abramsky2011sheaf}, our result still applies with Kochen and Specker's stricter original definition \cite{kochenspecker67}. In particular, our proof still holds if we limit KS contextuality to behaviors achievable with ideal measurements \cite{budroni2022contextualityReview}, because the KS-noncontextual and KS-contextual behaviors we use ($P^{\rm A},P^{(3)}$) have quantum models.
Our result also applies to continuous-valued measures of contextuality \cite{abramsky2017fraction,kujala2019measures,vallee2024corrected}, because the inconsistency holds for the set of behaviors assigned zero contextuality (i.e., those admitting fully classical explanations).

Importantly, we define post-processing to depend only on the measured values of the observables and not on the context, following other recent work \cite{shane2019comonadic,Karvonen2019categories,barbosa2021closing}. 
If post-processing were allowed to depend on context then it could generate any behavior and the concept of contextuality would be empty \cite{dzhafarov2023redundancy}. 
Our definition of post-processing matches the Kochen--Specker theorem, which is concerned with algebraic relations among observables of the form $q'\equiv f(q_1,\dots,q_n)$ \cite{kochenspecker67}.
Moreover, KS-noncontextual models allow each observable to depend on the hidden state of the system but not on which other observables are measured
\footnote{It was recently claimed that KS contextuality violates the post-processing axiom, contradicting our Theorem \ref{thm:KS-satisfies-axioms} \cite{dzhafarov2023redundancy}. However, most of the counterexamples involve functions that cannot be consistently defined on their full domains (see Eqs. 15, 25-26, 28 of Ref. \cite{dzhafarov2023redundancy}). In the remaining counterexample (Eq. 20), the function depends on which observables have been measured, making it effectively a function of context.}.

Our results are closely related to resource theories of contextuality. 
Resource theories specify a set of objects (resources), a subset of free objects, free operations that transform objects to objects, and real-valued functions called monotones that cannot increase under free operations \cite{coecke2016resourceTheory}. In resource theories of contextuality, the objects are nondisturbing behaviors, free objects are the KS-noncontextual ones, and monotones are measures of contextuality such as contextual fraction \cite{abramsky2017fraction}.
The free operations which cannot create or increase contextuality
include those invoked in
Principles \ref{ax: Nestedness}, \ref{ax: post-processing} and \ref{ax: independence} (Monotonicity, Post-processing, and Independence) \cite{abramsky2017fraction,shane2019comonadic,Karvonen2019categories,barbosa2021closing}. That is, each of these principles is of the form ``if $P$ (or $P^{(1)}$ and $P^{(2)}$) is noncontextual then so is $P'\,$'', and in every case the implied transformation is a free operation.
Our only proposed principle not implied by resource theories is Determinism.
Thus a corollary of our result is that any extension of contextuality obeying Principle \ref{ax: determinism} must violate extant resource theories, in that contextuality can be created from noncontextuality using free operations. 

We hope the present work will help point the way toward a 
theory of extended contextuality that 
maintains the strong no-go implications of KS contextuality while remaining internally consistent.
One route forward may be 
to take into account the physical properties of a measurement system, in contrast to the currently dominant model-independent approach.
Specifically, post-processings might be restricted
to somehow respect a system's compositional structure,
in particular with respect to causal effects of context (which are relevant only for disturbing systems).
For example, in the proof of Theorem \ref{thm:impossibility-a}, 
$P^{(1)}$ is constructed from the product of independent systems $P^{\rm A}\otimes P^{\rm B}$, with only $P^{\rm B}$ exhibiting disturbance. The post-processings defining $q_2$ and $q_3$ then combine information from the two subsystems so that the decomposition is no longer evident in $P^{(2)}$. Existing frameworks for contextuality do not have the expressive capacity for this sort of structural information, but it may be possible to develop richer theories that do.
The contrast between Bell nonlocality and KS contextuality is instructive here:
Bell inequalities do not require ideal measurement \cite{budroni2022contextualityReview} because they assume spacelike separation, which allows
each observer's measurement to depend on their own setting but not on the other observer's setting.
Post-processing preserves ideal measurement but it does not respect spacelike separation. Therefore Principle \ref{ax: post-processing} is not justified for Bell nonlocality, so Theorem \ref{thm:impossibility-a} does not generalize from KS contextuality to Bell nonlocality. Building on these observations might lead to a theory of extended contextuality accommodating more complex constraints on allowable dependencies between contexts and observables in specific physical scenarios. 
Constraints of this sort may be expressible in terms of a family of allowable causal graphs, generalizing recent graph-based approaches to contextuality \cite{wood2015tuning,cavalcanti2018tuning,Jones2019mcontextuality,pearl2021classical}, and contextuality might be defined as incompatibility of an empirical behavior with any hidden-variables model obeying those constraints. Unfortunately, current graph-based approaches either violate our principles \cite{Jones2019mcontextuality} or are arguably too brittle in that they classify any behavior deviating from certain exact conditional independencies as noncontextual \cite{wood2015tuning,cavalcanti2018tuning}. An open question is whether a theory of extended contextuality can incorporate structural information in a way that escapes the implications of our theorem yet remain robust to perturbations or sampling error.

\begin{acknowledgments}

A.T. acknowledges support from  the USP Center of Excellence in Quantum Science and Technology (USP/IFSC Process No.~26.1.14.1.7).  A.T. further acknowledges support from the National Council for Scientific and Technological Development (CNPq, Brazil) and Quantum Industrial Innovation (QuIIN)---EMBRAPII CIMATEC Competence Center in Quantum Technologies, with financial resources from the PPI IoT/Manufatura 4.0 of the MCTI (Grant No.~053/2023).

B.A. acknowledges support from the São Paulo Research Foundation (FAPESP) and the Serrapilheira Institute.

Research at Perimeter Institute is supported in part by the Government of Canada through the Department of Innovation, Science and Economic Development and by the Province of Ontario through the Ministry of Colleges and Universities.

This research was performed in part while M.J. was at Google Brain and Google DeepMind.

\end{acknowledgments}

\appendix
\renewcommand{\thesection}{\Alph{section}}

\section{Relations Among Principles}\label{appendix: relations}

This appendix presents some logical relationships by which certain of our proposed principles follow from others. We also describe three other principles that are closely related to (and implied by) the four in the main text. Some of the results here provide further motivation for the principles by showing they reduce to simpler ones (e.g., Lemma \ref{lemma: J+C>P}).

We use lowercase names for transformations of behaviors and capitalized terms for the corresponding principles. For example, post-processing is a type of transformation that appends a new observable determined by others, while Post-processing is the property that this transformation cannot create contextuality.

We formalized Monotonicity (Principle \ref{ax: Nestedness}) in terms of removing entire observables from certain contexts.
Another way of removing information is by \emph{coarse-graining} an observable \cite{kujala2021dichotomization}. Formally, coarse-graining involves replacing $q$ with $q'\equiv g(q)$ for some many-to-one function $g$. Given a behavior $P$, the coarse-grained behavior $P'$ is defined by $\mathcal{Q}'=\mathcal{Q}\cup\{q'\}\setminus\{q\}$ with $q'\prec' c$ whenever $q\prec c$ and $P'_{q'}(\out'|c) = \sum_{g(\out)=\out'}P_{q}(\out|c)$ (extended in the natural way to joint probabilities with other observables).
We argue coarse-graining should not create contextuality for the same reasons given in the main text for Monotonicity.

\begin{axiom}[Coarse-graining \cite{kujala2021dichotomization}]\label{ax: Coarse-graining}
Any coarse-graining of a noncontextual behavior is noncontextual.    
\end{axiom}

A special case of post-processing is when $f$ is injective, implying $q'$ can be identified with the joint outcome of $\{q_1,\dots,q_n\}$. This operation has been called \emph{joining} \cite{dzhafarov2017canonical}:
Given a behavior $P$ and a set of observables $\boldsymbol{q}=\{q_1,\dots,q_n\}\subset\mathcal{Q}$,
define an expanded behavior $P'$ by appending their joint outcome as a vector-valued observable $q' \equiv (q_{1},\dots,q_{n})$,
meaning $\mathcal{Q}'=\mathcal{Q}\cup\{q'\}$
with $q'\prec c$ whenever $\boldsymbol{q}\prec c$,
and $P'(q'=(q_1,\dots,q_n)|c)=1$.
Similarly to the rationale offered for Post-processing (Principle \ref{ax: post-processing}), joining introduces no new information and therefore should not create contextuality.

\begin{axiom}[Joining]\label{ax: Joining}
Given any noncontextual behavior and any set $\boldsymbol{q}$ of observables in the scenario, the expanded behavior obtained by joining the variables in $\boldsymbol{q}$ (i.e., incorporating their joint outcome as a new observable) is also noncontextual.
\end{axiom}

In the proof of Theorem \ref{thm:impossibility-a}, we invoke Post-processing and Monotonicity to relabel certain observables in certain contexts.  This procedure can be regarded as a transformation in its own right, with a corresponding \emph{Relabeling} principle that contextuality cannot be created by relabeling an observable differently in different contexts:

\begin{axiom}[Relabeling]\label{ax: Relabeling} Let $P$ be any behavior involving an observable $q$, and let $\mathcal{C}_{q}=\{c:q\prec c\}$ be the collection of all contexts containing $q$. Let $\{\mathcal{C}_{1},\dots,\mathcal{C}_{n}\}$ be any partition of $\mathcal{C}_{q}$, and relabel $q$ as $q_{i}$ in each $\mathcal{C}_{i}$, where $q_i$ represents a new observable not already in $\mathcal{Q}$. That is, let $\mathcal{Q}'=\mathcal{Q}\cup\{q_1,\dots,q_n\}\setminus\{q\}$, with $O_{q_{i}} = O_{q}$ and $q_{i} \prec' c$ iff $c \in \mathcal{C}_{i}$ for $i=1,\dots,n$. Finally, let $P'$ be the natural translation of $P$ into this relabeled scenario. If $P$ is noncontextual then so is $P'$.
\end{axiom}

We now present some logical relationships among the proposed principles, beginning with the observation that
joining is a special case of post-processing:

\begin{lemma}\label{lemma: P>J} Post-processing (P\ref{ax: post-processing}) $\implies$ Joining (P\ref{ax: Joining}) \end{lemma}
\begin{proof}
A composite observable $q'$ defined by joining a collection $\boldsymbol{q}=\{q_1,\dots,q_n\}$ is a post-processing with $f(\boldsymbol{q})=(q_1,\dots,q_n)$. Therefore, if Post-processing is satisfied by some extension of contextuality, Joining also is.
\end{proof}

A complement to Lemma \ref{lemma: P>J} is that post-processing reduces to
a composition of joining and coarse-graining (i.e., coarse-graining the composite variable):

\begin{lemma}\label{lemma: J+C>P} Joining (P\ref{ax: Joining}) + Coarse-graining (P\ref{ax: Coarse-graining}) $\implies$ Post-processing (P\ref{ax: post-processing}) \end{lemma}
\begin{proof}
Given a noncontextual behavior $P$ and a post-processing $f(q_1,\dots,q_n)$, define the composite variable $q'=(q_1,\dots,q_n)$. Principle \ref{ax: Joining} entails that appending $q'$ to $P$ yields a noncontextual behavior. Now define the function $g(\outs)\equiv f(u_1,\dots,u_n)$ so that $g(q')=f(q_1,\dots,q_n)$. Principle \ref{ax: Coarse-graining} entails that replacing $q'$ with $g(q')$ yields a noncontextual behavior. This behavior is the same as appending the post-processing $f(q_1,\dots,q_n)$ to $P$.
\end{proof}

Post-processing and coarse-graining are also related in that a coarse-graining is a post-processing of a single observable. Since we have defined post-processing as appending a new observable (and keeping the originals) while coarse-graining replaces the original observable, we must also include Monotonicity in the following lemma:

\begin{lemma}\label{lemma: N+P>C} Monotonicity (P\ref{ax: Nestedness}) + Post-processing (P\ref{ax: post-processing}) $\implies$ Coarse-graining (P\ref{ax: Coarse-graining}) \end{lemma}
\begin{proof}
Given a coarse-graining $q'\equiv g(q)$, define a post-processing of the singleton $\{q\}$ by $f(q)\equiv g(q)$. Principle \ref{ax: post-processing} entails that appending $q'$ to a noncontextual behavior $P$ yields a noncontextual $P'$. Principle \ref{ax: Nestedness} then entails that dropping the original $q$ yields another noncontextual behavior $P''$, which is the coarse-graining of $P$.
\end{proof}

Although relabeling does not fit the definition of post-processing because it depends on context, Relabeling (Principle \ref{ax: Relabeling}) is implied by Post-processing and Monotonicity:

\begin{lemma}\label{lemma: N+P>R} Monotonicity (P\ref{ax: Nestedness}) + Post-processing (P\ref{ax: post-processing}) $\implies$ Relabeling (P\ref{ax: Relabeling}) \end{lemma}
\begin{proof}
Let $q$ be an observable of some noncontextual behavior $P$ and let $\{\mathcal{C}_{1},\dots,\mathcal{C}_{n}\}$ be any partition of $\mathcal{C}_{q}=\{c:q\prec c\}$. For each $i=1,\dots,n$, let $q_{i}$ be a copy of $q$ defined by post-processing, using the identity function $f(u)\equiv u$. Principle \ref{ax: post-processing} entails that the expanded behavior obtained by incorporating every $q_{i}$ is noncontextual. Principle \ref{ax: Nestedness} then entails that, if we remove each $q_{i}$ from all contexts $c \notin \mathcal{C}_{i}$ (i.e., imposing $q_{i} \prec' c$ if and only if $c \in \mathcal{C}_{i}$) and drop $q$ altogether, the resulting behavior remains noncontextual. It is easy to see that this behavior corresponds to relabeling $q$ as prescribed in Principle \ref{ax: Relabeling}.
\end{proof}

Lastly, we note that Determinism 
is the only principle that provides a logical connection from nondisturbing to disturbing systems. That is, it is easy to show that restrictions, post-processings, products, coarse-grainings, joinings and relabelings of nondisturbing behaviors are all nondisturbing. Therefore, given the KS contextuality of all nondisturbing behaviors, Principles \ref{ax: Nestedness}-\ref{ax: Relabeling} taken alone have no implications for disturbing ones. In contrast, deterministic behaviors
can be disturbing. Hence, Determinism is the only principle that prevents a trivial extension of contextuality whereby all disturbing systems are classified as contextual (or noncontextual). 

\renewcommand{\thesection}{\Alph{section}}
\section{Proof of Theorem \ref{thm:KS-satisfies-axioms}}\label{appendix: theorem-proof}

Our proof of Theorem \ref{thm:KS-satisfies-axioms} uses noncontextual models as defined by Kochen and Specker \cite{kochenspecker67}. 
We write a noncontextual model as $\mathcal{M}=(\boldsymbol{\Lambda},\mathcal{F})$, where $\boldsymbol{\Lambda}=(\Lambda,\Sigma,\Pr)$ is a probability space and $\mathcal{F}$ is a family of measurable random variables $F_q:\Lambda\to O_q$. 
We say $\mathcal{M}$ is a model for a behavior $P$ if they agree on the joint distribution of observables in each context: $P(\cdot|c) = \Pr(\mathcal{F}_{\boldsymbol{c}})$ for all $c\in\mathcal{C}$, where $\mathcal{F}_{\boldsymbol{c}}=\{F_q:q\prec c\}$.
It is well known that such a model is equivalent to a global probability assignment \cite{fine1982hidden,abramsky2011sheaf}.

For Principle \ref{ax: determinism} (Determinism), the proof directly constructs a noncontextual model for any nondisturbing deterministic behavior. The other principles are all of the form that if behavior $P$ (or $P^{(1)}$ and $P^{(2)}$) is noncontextual then so is some $P'$. In each case, the proof assumes a noncontextual model for $P$ and constructs one for $P'$. 
Note that in all cases, if $P$ is nondisturbing then $P'$ necessarily is too.

\textbf{Determinism (Principle \ref{ax: determinism})}:
Given a deterministic and nondisturbing behavior $P$ with values $\{u_q:q\in\mathcal{Q}\}$ satisfying $P_q(u_q\vert c)=1$ whenever $q\prec c$, there is a trivial noncontextual model for it with a single hidden state: $\Lambda=\{\lambda\}$, $\Pr[\lambda]=1$, $\forall_q F_q(\lambda)=u_q$.

For Principles \ref{ax: Nestedness} and \ref{ax: post-processing}, assume we are given a behavior $P$ with a noncontextual model $\mathcal{M}=(\boldsymbol{\Lambda},\mathcal{F})$.

\textbf{Nestedness (Principle \ref{ax: Nestedness})}:
$\mathcal{M}$ is automatically a model for any restriction of $P$. Indeed, the assumption $P(\cdot|c) = \Pr(\mathcal{F}_{\boldsymbol{c}})$ implies $P_{\boldsymbol{c}'}(\cdot|c) = \Pr(\mathcal{F}_{\boldsymbol{c}'})$ for any $\boldsymbol{c}'\subset\boldsymbol{c}$.

\textbf{Post-processing (Principle \ref{ax: post-processing})}:
Let $P'$ be the expansion of $P$ obtained by incorporating any post-processing $q'\equiv f(\boldsymbol{q})$ of a collection $\boldsymbol{q} =\{q_{1},\dots,q_{n}\} \subset \mathcal{Q}$. Define $F_{q'}: \Lambda \rightarrow O_{q'}$ by $F_{q'}(\lambda) = f(F_{q_1}(\lambda),\dots,F_{q_n}(\lambda))$. Then $\mathcal{M}'=(\boldsymbol{\Lambda},\mathcal{F}\cup\{F_{q'}\})$ is a noncontextual model for $P'$.

\textbf{Independence (Principle \ref{ax: independence})}:
Let $P^{(1)}$ and $P^{(2)}$ be behaviors with noncontextual models $\mathcal{M}_{1} = (\boldsymbol{\Lambda}_{1},\mathcal{F}_{1})$ and $\mathcal{M}_{2} = (\boldsymbol{\Lambda}_{2},\mathcal{F}_{2})$. In the product space $\boldsymbol{\Lambda} \equiv (\Lambda_{1} \times \Lambda_{2},\Sigma_{1} \otimes \Sigma_{2},\Pr_1\times\Pr_2)$, define $F^{\prime}_q$ for each $q\in\mathcal{Q}_k$ ($k=1,2$) by $F^{\prime}_q(\lambda_{1},\lambda_{2}) = F_{k,q}(\lambda_{k})$. It is easy to see that $\mathcal{M}'=(\boldsymbol{\Lambda},\mathcal{F}')$ is a noncontextual model for $P^{(1)} \times P^{(2)}$. \hfill $\square \medskip $

Note that Theorem \ref{thm:KS-satisfies-axioms} also holds for \textbf{Coarse-graining, Joining and Relabeling (Principles \ref{ax: Coarse-graining}, \ref{ax: Joining} and \ref{ax: Relabeling})},
due to Lemmas \ref{lemma: N+P>C}, \ref{lemma: P>J} and \ref{lemma: N+P>R}, respectively.

\section{On the Generality of the Inconsistency}
\label{appendix: generality}

Theorem \ref{thm:impossibility-a} shows no extension of contextuality can be consistent with our proposed principles, but one might ask how localized the obstruction is. Is it limited to some class of pathological behaviors such that if they were ignored one could consistently define contextuality on all others?
(An obvious reply is that the pathological class would have to include $P^{(3)}$, which plays a central role in the theory of contextuality.)
We prove the answer is negative, using an observation that there is a sequence of free operations leading from any behavior to a KS-contextual behavior.

Formally, define a \emph{ternary extension of contextuality} as one that labels every behavior as $\texttt{contextual}$, $\texttt{noncontextual}$ or $\texttt{undefined}$ and that agrees with KS contextuality for nondisturbing behaviors.
Our principles can be applied to such extensions without modification.

\begin{cor}\label{cor: generality}
Any ternary extension of contextuality satisfying Principles \ref{ax: determinism}, \ref{ax: Nestedness}, \ref{ax: post-processing} \& \ref{ax: independence} must label every behavior as contextual or undefined.
\end{cor}

\begin{proof}
Let $P^{(0)}$ be any behavior. Define $P^{\rm A}$ and $P^{\rm B}$ as in the proof of Theorem \ref{thm:impossibility-a}. Form $P^{(0)}\otimes P^{\rm A}\otimes P^{\rm B}$ and define the post-processings $q_2$ and $q_3$ as in that proof. The resulting behavior has a restriction which after appropriate relabeling matches $P^{(3)}$. Therefore if $P^{(0)}$ were labeled noncontextual then so would $P^{(3)}$, a contradiction.    
\end{proof}


\bibliographystyle{apsrev4-2}
\bibliography{bibliography}

\end{document}